\documentclass[a4paper,UKenglish, numberwithinsect]{lipics-v2018}
\hideLIPIcs
\nolinenumbers
\usepackage{graphicx}
\usepackage{float}

\usepackage{algorithm}
\usepackage{algpseudocode}% http://ctan.org/pkg/algorithmicx

\usepackage{array,multirow,makecell}

\usepackage{tikz}
\usetikzlibrary{shapes, arrows}

\usepackage[color=green!30]{todonotes}
\usepackage{xspace}	%a space after a macro, only if it's not a ponctuation sign

\usepackage{comment}

\usepackage{booktabs}
\usepackage{scrpage2}
\theoremstyle{plain}\newtheorem{proposition}[theorem]{Proposition}
\theoremstyle{plain}\newtheorem{rrule}{Reduction Rule}

\usepackage{chngcntr}

\counterwithin*{equation}{section}
\counterwithin*{equation}{subsection}
\counterwithin*{equation}{subsubsection}

\theoremstyle{plain}

\newcommand{\np}{\mathsf{NP}}
\newcommand{\wone}{\mathsf{W[1]}}
\newcommand{\wtwo}{\mathsf{W[2]}}
\newcommand{\fpt}{\mathsf{FPT}}

\newcommand{\p}{\mathsf{P}}
\newcommand{\bigo}{\mathcal{O}}
\newcommand{\bigos}{\mathcal{O}^*}

\newcommand{\cnpp}{\mathsf{coNP/Poly}}

\renewcommand{\c}[0]{\mathcal{C}\xspace}
\renewcommand{\t}{\mathcal{T}\xspace}
\newcommand {\h}{\mathcal{H}\xspace}
\newcommand {\s}{\mathcal{S}\xspace}
\renewcommand {\u}{\mathcal{U}\xspace}

\newcommand {\nhs}{n_{\mathcal{H}}^*\xspace}
\newcommand{\ac}{A_{\c}}
\newcommand{\tw}{\h_{t}}
\newcommand{\lc}{\ell_{\c}\xspace}

\DeclareMathOperator{\col}{col}

\def \mcsT {\textsc{Maximum Colorful Subtree}\xspace}
\def \mcaT {\textsc{Maximum Colorful Arborescence}\xspace}

\def \mca {\textsc{MCA}\xspace}

\def \msc {\textsc{Multicolored Set Cover}\xspace}
\def \sc {\textsc{Set Cover}\xspace}

\title{On the Maximum Colorful Arborescence Problem and Color Hierarchy Graph Structure\footnote{This work was partially supported by PHC PROCOPE number 37748TL and by the DFG, project MAGZ (KO 3669/4-1)}}
\titlerunning{On the Maximum Colorful Arborescence Problem and Color Hierarchy Graph Structure}

\author{Guillaume Fertin}{LS2N UMR CNRS 6004, Universit\'e de Nantes, Nantes, France}{guillaume.fertin@univ-nantes.fr}{}{}
\author{Julien Fradin}{LS2N UMR CNRS 6004, Universit\'e de Nantes, Nantes, France}{julien.fradin@univ-nantes.fr}{}{}
\author{Christian Komusiewicz}{Fachbereich für Mathematik und Informatik, Philipps-Universität Marburg, Marburg, Germany}{komusiewicz@informatik.uni-marburg.de}{}{}
\authorrunning{G. Fertin, J. Fradin and C. Komusiewicz} %mandatory. First: Use abbreviated first/middle names. Second (only in severe cases): Use first author plus 'et. al.'

\Copyright{Guillaume Fertin, Julien Fradin and Christian Komusiewicz}%mandatory, please use full first names. LIPIcs license is "CC-BY";  http://creativecommons.org/licenses/by/3.0/

% mandatory: Please choose ACM 1998 classifications from http://www.acm.org/about/class/ccs98-html . E.g., cite as "F.1.1 Models of Computation". 
% mandatory: Please provide 1-5 keywords
\begin{document}
\maketitle
\begin{abstract}
Let $G=(V,A)$ be a vertex-colored arc-weighted directed acyclic graph (DAG) rooted in some vertex $r$. The color hierarchy graph $\h(G)$ of $G$ is defined as follows: $V(\h(G))$ is the color set $\c$ of $G$, and $\h(G)$ has an arc from $c$ to $c'$ if $G$ has an arc from a vertex of color $c$ to a vertex of color $c'$. We study the \mcaT (MCA) problem, which takes as input a DAG $G$ such that $\h(G)$ is also a DAG, and aims at finding in $G$ a maximum-weight arborescence rooted in $r$ in which no color appears more than once. 
The \mca problem models the {\em de novo} inference of unknown metabolites by mass spectrometry experiments. %(see e.g.~\cite{DBLP:conf/eccb/BockerR08}). 
Although the problem has been introduced ten years ago (under a different name), it was only recently pointed out that a crucial additional property in the problem definition was missing: by essence, $\h(G)$ {\em must be} a DAG. In this paper, we further investigate \mca under this new light and provide new algorithmic results for this problem, with a specific focus on fixed-parameter tractability ($\fpt$) issues for different structural parameters of $\h(G)$.
In particular, we show there exists an $\bigos(3^{\nhs})$ time algorithm for solving \mca, where $\nhs$ is the number of vertices of indegree at least two in $\h(G)$, thereby improving the $\bigos(3^{|\c|})$ algorithm from Böcker et al.~[Proc.~ECCB~'08].
We also prove that \mca is $\wtwo$-hard relatively to the treewidth $\tw$ of the underlying undirected graph of $\h(G)$, and further show that it is $\fpt$ relatively to $\tw + \lc$, where $\lc := |V| - |\c|$.
\end{abstract}

\section{Introduction}
Motivated by {\em de novo} inference of metabolites from mass spectrometry experiments, Böcker et al.~\cite{DBLP:conf/eccb/BockerR08} introduced the \mcsT problem, an optimization problem that takes as input a vertex-colored arc-weighted directed acyclic graph $G=(V,A)$ rooted in some vertex $r$, and asks for a maximum weighted arborescence in $G$ that contains $r$, and in which each color appears at most once. In this model, the root $r$ in $G$ represents the sought metabolite, any vertex in $G$ represents a molecule obtained from $r$ after (possibly several) fragmentation(s), and vertices are colored according to their masses. An arc connects two molecules (vertices) $u$ and $v$ when $v$ can be obtained from $u$ by fragmentation, and is assigned a weight that indicates the (possibly negative) degree of confidence that the fragmentation from $u$ to $v$ actually occurs. A maximum weighted arborescence from $G$ that contains $r$ and in which each color appears at most once thus represents a most plausible fragmentation scenario from $r$. Let $\h(G)$ be the following graph built from $G$: $V(\h(G))$ is the set $\c$ of colors used to color $V(G)$, and there is an arc from $c$ to $c'$ in $\h(G)$ if there is an arc in $G$ from a vertex of color $c$ to a vertex of color $c'$. We call $\h(G)$ the {\em color hierarchy graph} of $G$. Observe that $\h(G)$ {\em must be} a DAG since colors represent masses and fragmenting a molecule gives new molecules with lower mass. As recently pointed out~\cite{DBLP:conf/tamc/FertinFJ17}, the initial definition of \mcsT omits this crucial property of $G$.  This led Fertin et al.~\cite{DBLP:conf/tamc/FertinFJ17} to reformulate the initial \mcsT problem as follows.

\begin{center}
\fbox{\begin{minipage}{0.95\textwidth}
\textbf{\mcaT (\mca)}\\
\textbf{Input:} A DAG $G = (V,A)$ rooted in some vertex $r$, a set $\c$ of colors, a coloring function $\col : V \to \c$ such that $\h(G)$ is a DAG and an arc weight function $w : A \to \mathbb{R}$.\\
\textbf{Output:} A colorful arborescence $T = (V_{T}, A_{T})$ rooted in $r$ and of maximum weight $w(T) = \sum_{a \in A_{T}} w(a)$. 
\end{minipage}}
\end{center}

The study of \mca initiated in~\cite{DBLP:conf/tamc/FertinFJ17} essentially focused on the particular case where $G$ is an arborescence (\emph{i.e.} the underlying undirected graph of $G$ is a tree), and showed for example that \mca is $\np$-hard even for very restricted such instances. This work was also the first one to explicitly exploit that $\h(G)$ is a DAG. In particular, it was shown that if $\h(G)$ is an arborescence, then \mca is polynomially solvable. This latter promising result is the starting point of the present paper, in which we aim at better understanding the structural parameters of $\h(G)$ that could lead to fixed-parameter tractable ($\fpt$), \emph{i.e.} exact and moderately exponential, algorithms. As pointed out in a recent study~\cite{BoeckerArxiv18}, obtaining exact solutions instead of approximate ones is indeed preferable for \mca. Hence, improved exact algorithms are truly desirable for this problem.

\paragraph*{Related work and our contribution}
The \mca problem is $\np$-hard even when every arc weight is equal to 1~\cite{DBLP:conf/eccb/BockerR08} and highly inapproximable even when $G$ is an arborescence with uniform weights~\cite{DBLP:conf/tamc/FertinFJ17}. Moreover, \mca is  $\wone$-hard parameterized by $\lc =|V(G)| - |\c|$~\cite{DBLP:conf/tamc/FertinFJ17} (a consequence of Theorem 1 from~\cite{9410444820120101}). On the positive side, \mca can be solved  in $\bigo^*(3^{|\c|})$ time by dynamic programming~\cite{DBLP:conf/eccb/BockerR08}. Moreover, as previously mentioned, \mca is in $\p$ when $\h(G)$ is an arborescence~\cite{DBLP:conf/tamc/FertinFJ17}. This result can be extended to some arborescence-like color hierarchy graphs as \mca can be solved by a branching algorithm in time $\bigos(2^s)$~\cite{DBLP:conf/tamc/FertinFJ17} where $s$ is the minimum number of arcs of $\h$ whose removal turns $\h$ into an arborescence. Finally, a solution of \mca of order $k$ can be computed in $\bigos((3e)^k)$ time using the color-coding technique~\cite{512234419950701} in combination with dynamic programming~\cite{1093576520090101}.

A related pattern matching problem in graphs is \textsc{Graph Motif} where, in its simplest version, we are given an undirected vertex-colored graph and ask whether there is a connected subgraph containing one vertex of each color~\cite{lacroix,fellowsupper,BBFKN11,BKK16}. The main difference is that the graph does not contain edges of negative weight. As a consequence, \textsc{Graph Motif} is somewhat simpler and, in contrast to \mca, \textsc{Graph Motif} is fixed-parameter tractable for the parameter~$\lc$~\cite{BBFKN11,DBLP:conf/cpm/FertinK16}. 
Our results are summarized in Table~\ref{tab:results}. We focus on two parameters from $\h(G)$, namely its number $\nhs$ of vertices of indegree at least 2, and the treewidth $\tw$ of its underlying undirected graph. This choice is motivated by the fact that when $\h(G)$ is an arborescence, each of these two parameters is constant (namely, $\nhs=0$ and $\tw=1$) while \mca is in $\p$. Thus, our parameters measure the distance from this trivial case~\cite{GHN04}.  In addition, we consider the parameter $\lc := |V(G)| - |\c|$ which is the number of vertices that are not part of the solution even if the solution contains one vertex from each color, and $\ell\geq \lc$ which is the number of vertices that are not part of the solution. Together with $\fpt$ issues, we also address the (in)existence of polynomial problem kernels for these parameters. In a nutshell, we provide an almost complete dichotomy for fixed-parameter tractability and problem kernelization for these parameters. 

\begin{table}[t]
%\hspace{-0.55cm} %Décalage tableau
\footnotesize

\caption{Overview of the results for the \mca problem presented in this paper. Here, $\nhs$ is the number of vertices of indegree at least 2 in $\h$, $\tw$ is the treewidth of the underlying undirected graph of $\h$, $\lc := |V(G)| - |\c|$ and $\ell\geq \lc$ is the number of vertices that are not part of the solution.}
\centering
\begin{tabular}{|c||c|c|}
\hline
Parameter & $\fpt$ status & Kernel status \\ 
\hline
\hline
$\nhs$ & ~ $\bigos(3^{\nhs})$ (Th.~\ref{thm:fpt-in2}) ~ & ~ No poly. kernel (Prop.~\ref{prop:no-poly-kernel-c})~ \\
\hline
$\ell$ & \multicolumn{2}{c|}{$\wone$-hard (from \cite{9410444820120101})} \\ 
\hline
$\nhs + \lc$ & $\fpt$ (from Th.~\ref{thm:fpt-in2}) & No poly. kernel (Prop.~\ref{prop:no-poly-kernel-lc-in2}) \\
\hline
$\nhs + \ell$ & \multicolumn{2}{c|}{Poly. kernel (Th.~\ref{thm:kernel-l-in2})} \\
\hline
$\tw$ & \multicolumn{2}{c|}{$\wtwo$-hard (Prop.~\ref{prop:w2-tw})} \\
\hline
$\tw + \lc$ & $\bigos(2^{\lc}\cdot 4^{\tw})$ (Th.~\ref{thm:fpt-lc-twch}) & ? \\
\hline
\end{tabular}
\label{tab:results}
\end{table}
\paragraph*{Preliminaries}
In the following, let $G = (V,A)$ be the input graph of \mca, with $n_G := |V(G)|$. For any integer $p$, we let $[p] := \{1,\ldots,p\}$. For any vertex $v \in V$, $N^+(v)$ is the set of outneighbors of $v$. A set~$S \subseteq V$ (resp. a graph $G$) is  \textit{colorful} if no two vertices in $S$ (resp. in $G$) have the same color. Moreover, we say that a subgraph $G'$ of $G$ is \textit{fully-colorful} if it contains exactly one occurrence of each color from $\c$. The color hierarchy graph of $G$ is denoted $\h(G) :=(\c, \ac)$, or, when clear from the context, simply $\h$. For any instance of \mca, we define $\lc := n_G - |\c|$ and we denote by $\ell$ the number of vertices that are not part of the solution -- thus $\ell \geq \lc$.
We finally briefly recall the relevant notions of parameterized
algorithmics (see e.g.~\cite{DBLP:books/sp/CyganFKLMPPS15}). A parameterized problem is a subset of $\Sigma\times \mathbb{N}$ where the second component is the parameter. A parameterized problem is \emph{fixed-parameter tractable} if every instance~$(x,k)$ can be solved in~$f(k)\cdot |x|^{O(1)}$ time. A \emph{reduction to a problem
  kernel}, or \emph{kernelization}, is an algorithm that takes as
input an instance~$(x,k)$ of a parameterized problem $Q$ and produces in
polynomial time an equivalent (\emph{i.e.}, having the same solution) instance~$(x',k')$ of~$Q$ such that (i)~$|x'|\le g(k)$, and (ii)~$k'\le k$.  The
instance~$(x',k')$ is called \emph{problem kernel}, and~$g$ is called
the \emph{size of the problem kernel}. If~$g$ is a polynomial
function, then the problem admits a \emph{polynomial-size
  kernel}. Classes $\wone$ and $\wtwo$ are classes of presumed
fixed-parameter intractability: if a problem is
$\wone$-hard (resp. $\wtwo$-hard) for parameter~$k$, then it is generally assumed that it is not fixed-parameter tractable.

This paper is organized as follows. In Section~\ref{sec:in2}, we study in detail the impact of $\nhs$ on the parameterized complexity of the \mca problem, while in Section~\ref{sec:tw}, the same type of study is realized with parameter $\tw$. %In both sections, we provide positive and negative results, which help understand the border between tractable and intractable cases. 
%Section~\ref{sec:concl} is our conclusion.\todo{GF: due to lack of space, I'm not sure we'll keep the conclusion in the end}
Due to lack of space, some proofs are deferred to the appendix.
% (Section~\ref{sec:appendix}).

\section{Parameterizing the \mca Problem by $\nhs$}
\label{sec:in2}
Two main reasons lead us to be particularly interested in $\nhs$. First, \mca is in $\p$ when $\h$ is an arborescence~\cite{DBLP:conf/tamc/FertinFJ17}, thus when $\nhs = 0$. Second, \mca can be solved in $\bigos(3^{|\c|})$ time~\cite{DBLP:conf/eccb/BockerR08}. Since by definition $\nhs \leq |\c|$, determining whether \mca is $\fpt$ with respect to $\nhs$ is of particular interest. We answer this question positively in Theorem~\ref{thm:fpt-in2}.
We first need some additional definitions. Let $X$ be the set of vertices of indegree at least two in $\h$ (thus $|X| = \nhs$) and call it the set of \textit{difficult colors}. For any $V'\subseteq V(G)$, $\col(V')$ denotes the set of colors used by $\col$ on the vertices in $V'$. Moreover, for any vertex $v \in V$ that has at least one outneighbor in $G$, assume that $\col(N^+(v))$ has an arbitrary but fixed ordering of its vertices. Therefore, for any $i \in [|\col(N^+(v))|]$, let $\col^+(v, i)$ denote the $i$-th color in $\col(N^+(v))$. Finally, for any arborescence $T$ in $G$, let $X(T)$ denote the set of difficult colors in $\col(T)$. We have the following lemma.

\begin{lemma}
\label{lem:in2-disjoint-colors}
In $\h$, for any $c \in \c$, any pair of distinct colors $c_1, c_2 \in N^+(c)$ and any disjoint sets $X_1, X_2 \subseteq X$, any arborescence $T_1$ rooted in $c_1$ such that $X(T_1) \subseteq X_1$ is disjoint from any arborescence $T_2$ rooted in $c_2$ such that $X(T_2) \subseteq X_2$.
\end{lemma}

\begin{proof}
Assume wlog that $\h$ does not contain any path from $c_2$ to $c_1$. If $T_1$ and $T_2$ are not disjoint then there exists $c^* \in \c$ such that $c^*$ belongs both to $T_1$ and $T_2$. In order to prove that such a color $c^*$ cannot exist, let $\tau_1$ (resp. $\tau_2$) be the set of colors on the path from $c_1$ (resp. $c_2$) to $c^*$ including $c_1$~in~$T_1$ (resp. $c_2$~in~$T_2$). Then, either $\tau_2 \subset \tau_1$ or $c_2 \notin \tau_1$. First, if $\tau_2 \subset \tau_1$, then there exists a vertex $c' \in \tau_1$ such that $c' \neq c$ with an arc $(c',c_2)$. Since there already exists an arc $(c, c_2)$, $c_2$ is thus a difficult color, a contradiction to the assumption that $X_1$ and~$X_2$ are disjoint. Second, if $c_2 \notin \tau_1$, then $|\tau_1 \cap \tau_2| \geq 1$ since $c^* \in \tau_1 \cap \tau_2$. Therefore, let $\bar{c} \in \tau_1 \cap \tau_2$ such that there exists a path from $\bar{c}$ to any other color of $\tau_1 \cap \tau_2$. By definition, the father of $\bar{c}$ in $\tau_1$ is different from the father of $\bar{c}$ in $\tau_2$, which means that $\bar{c}$ is a difficult color, and thus contradicts the assumption that $X_1$ and $X_2$ are disjoint. 
\end{proof}

%Lemma~\ref{lem:in2-disjoint-colors} now leads to the following theorem.

\begin{theorem}
\label{thm:fpt-in2}
\mca can be solved in $\bigos(3^{\nhs})$ time and $\bigos(2^{\nhs})$ space.
\end{theorem}

\begin{proof}
We propose a dynamic programming algorithm which makes use of two programming tables.
The first one, $A[v, X', i]$, is computed for all $v \in V(G)$, $X' \subseteq X$ and $i \in \{0, \ldots , |\col(N^+(v))|\}$ and stores the weight of the maximum colorful arborescence $T_{A}(v, X', i)$ such that
\begin{itemize}
\item $T_{A}(v, X', i)$ is rooted in $v$,
\item $(X(T_{A}(v, X', i)) \setminus \{\col(v)\}) \subseteq X'$, and
\item $T_{A}(v, X', i)$ contains an arc $(v,u)$ only if $u \in N^+(v)$ and $\col(u) \in \underset{j \in [i]}{\cup}\col^+(v, j)$ in $\h$.
\end{itemize}
The second one, $B[v, X', i]$, is computed for all $v \in V$, $X' \subseteq X$ and $i \in [|\col(N^+(v))|]$ and stores the weight of the maximum colorful arborescence $T_{B}(v, X', i)$ such that
\begin{itemize}
\item $T_B(v, X', i)$ is rooted in $v$,
\item $(X(T_{B}(v, X', i)) \setminus \{\col(v)\}) \subseteq X'$, and
\item $T_B(v, X', i)$ contains an arc $(v,u)$ only if $u \in N^+(v)$ and $\col(u) = \col^+(v, i)$.
\end{itemize}
In a nutshell, $T_{A}(v, X', i)$ and $T_{B}(v, X', i)$ share the same root $v$ and the same allowed set of difficult colors $X'$ (disregarding $\col(v)$), but there cannot exist $u \in N^+(v)$ such that $(v,u) \in T_{A}(v, X', i-1)$ and $(v,u) \in T_{B}(v, X', i)$. We now show how to compute the two abovementioned tables.

\[
A[v, X', i] = 
\left\{ 
\begin{array}{l l}
  0 & \quad \text{if } i = 0 \\
  \underset{\forall X'' \subseteq X'}{\text{max}} \{ A[v, X'', i-1] + B[v, X' \setminus X'', i] \} & \quad \text{otherwise} \\ 
\end{array} 
\right.
\]

For an entry of type $A[v, X', i]$ with $i = 0$ recall that $T_A(v, X', i)$ can only contain $v$. Otherwise, observe that by definition there cannot exist any $u \in N^+(v)$ such that $u$ belongs both to $T_{A}(v, X'', i-1)$ and $T_{B}(v, X' \setminus X'', i)$. Therefore, Lemma~\ref{lem:in2-disjoint-colors} shows that $v$ is the only possible common vertex between $T_{A}(v, X'', i-1)$ and $T_{B}(v, X' \setminus X'', i)$ and thus that $T_{A}(v, X', i)$ is an arborescence. Moreover, for any $v \in V$ and any pair of vertices $\{u_1, u_2\} \subseteq N^+(v)$, Lemma~\ref{lem:in2-disjoint-colors} also shows that the color hierarchy graph of any arborescence rooted in $u_1$ is disjoint from the color hierarchy graph of any arborescence rooted in $u_2$, which proves that the combination of $A[v, X'', i-1]$ and $B[v, X' \setminus X'', i]$ is colorful. Finally, testing every combination of $X'' \subseteq X'$ ensures the correctness of the formula.

\[
B[v, X', i] = 
\left\{ 
\begin{array}{l l}
  \underset{\forall u~:~ \col(u) = \col^+(v,i) }{\text{max}} \{0, w(v,u) + A[u, X', |col(N^+(u))|] \} \\
  \phantom{0} \qquad \qquad \qquad \qquad \qquad \qquad \qquad \text{if $\col^+(v, i) \notin X$}\\
  
  \underset{\forall u~:~ \col(u) = \col^+(v,i) }{\text{max}} \{0, w(v,u) + A[u, X' \setminus \col(u), |\col(N^+(u))|] \} \\
  \phantom{0} \qquad \qquad \qquad \qquad \qquad \qquad \qquad \text{if $\col^+(v, i) \in X'$}\\ 
  
  0 \qquad \qquad \qquad \qquad \qquad \qquad \qquad \text{if $\col^+(v, i) \in X \setminus X'$}\\
\end{array} 
\right.
\]

For an entry of type $B[v, X', i]$, if $\col^+(v,i)$ is a difficult color which does not belong to $X'$, then $V(T_B(v, X', i)) = \{v\}$, and hence $B[v, X', i] = 0$. Otherwise, recall that $B[v, X', i]$ stores the weight of the maximum colorful arborescence rooted in a vertex $u \in N^+(v)$ which has color $\col^+(v,i)$ in addition to the weight $w(v,u)$. Therefore, computing the maximum colorful arborescences for any such $u$ and only keeping the best one if it is positive ensures the correctness of the formula. Finally, if $\col(u) \in X'$ then observe that $\col(u)$ cannot be contained a second time in $T_{B}(u, X', |\col(N^+(u))|)$ and must be removed from $X'$.

Recall that any DAG has a topological ordering of its vertices, \textit{i.e.} a linear ordering of its vertices such that for every arc $(u,v)$, $u$ appears before $v$ in this ordering. In Algorithm~\ref{algo1}, we show how, and in which order, to compute all the entries of both dynamic programming tables. For this, we consider the entries from last to first according to some topological ordering of $G$. The total running time derives from the fact that our algorithm needs $3^{\nhs}$ steps to compute $A[v, X', i]$ since a difficult color can be in $X''$, $X' \setminus X''$ or in $X \setminus X'$. 
\end{proof}

\begin{algorithm}[t]
\caption{\textsc{Computing the entries in tables $A$ and $B$}}
\label{algo1}
\begin{algorithmic}[t]
\footnotesize
\ForAll {$v \in V$ from last to first in some topological ordering of $G$}
    \ForAll {$X' \subseteq X$}
	  \ForAll {$i \in \{1, \ldots , |\col(N^+(v))|\}$}
		  \State Compute $B[v, X', i]$
	  \EndFor
    \EndFor
    \ForAll {$X' \subseteq X$}
	  \ForAll {$i \in \{0, \ldots ,|\col(N^+(v))|\}$}
		  \State Compute $A[v, X', i]$
	  \EndFor
    \EndFor
\EndFor
\end{algorithmic}
\end{algorithm}

%%%%%%%%%%%%%%%%%%%%%%%%%%%%%%%%%%%%%%%%%%%%%%%%%%%%%%%%%%%%%%%%%%%%%

%\subsection*{OR-comp for C}
Recall that a parameterized problem $Q$ is $\fpt$ with respect to a parameter~$k$ if and only if it has a kernelization algorithm for $k$~\cite{DBLP:series/txcs/DowneyF13}, but that such a kernel is not necessarily polynomial. In Proposition~\ref{prop:no-poly-kernel-c}, we prove that although  \mca is $\fpt$ relatively to $\nhs$ (as proved by Theorem~\ref{thm:fpt-in2}), \mca is unlikely to admit a polynomial kernel relatively to $\nhs$. For this, we use the or-cross composition technique which, roughly speaking, is a reduction that combines many instances of a problem into one instance of the problem $Q$. Hence, if an $\np$-hard problem admits an or-cross composition into a parameterized problem $Q$, then $Q$ does not admit any polynomial-size problem kernel (unless $\np \subseteq \cnpp$)~\cite{DBLP:journals/siamdm/BodlaenderJK14}. The or-cross composition we use actually shows that \mca is unlikely to admit a polynomial kernel relatively to $|\c|$, and consequently to $\nhs$.

\begin{proposition}
\label{prop:no-poly-kernel-c}
Unless $\np \in \cnpp$, \mca does not admit a polynomial kernel for parameter $|\c|$, and consequently for parameter $\nhs$, even if~$G$ is an arborescence.
\end{proposition}

The proof uses the notion of or-cross composition, that we first formally define.

\renewcommand{\thetheorem}{2.7}
\begin{definition} (\cite{DBLP:journals/jcss/BodlaenderDFH09,DBLP:journals/siamdm/BodlaenderJK14})
\label{def:comp-algo}
A \emph{composition algorithm} for a parameterized problem $Q \in \Sigma \times \mathbb{N}$ is an algorithm that receives as input a sequence $(x_1,k),(x_2,k), \ldots,(x_t,k)$ with $(x_i,k) \in \Sigma \times \mathbb{N}$ for each $1 \leq i \leq t$, takes polynomial time in $\sum_{i=1}^{t}|x_i|+k$, and outputs $(y,k') \in \Sigma \times \mathbb{N}$ with $(y,k') \in Q$ iff $\exists_{1 \leq i \leq t}(x_i,k) \in Q$ and $k'$ is polynomial in $k$. %A parameterized problem is called \emph{compositional} if there is a composition algorithm for it.
\end{definition}

\begin{proof}
In the following, let $t$ be a positive integer. For any $i \in [t]$, let $G_i=(V_i, A_i)$ be the graph of an instance of \mca which is rooted in a vertex $r_i$ and assume that the $t$~instances are built on the same color set $\c'=\{c_1,\ldots, c_{|\c'|}\}$. Moreover, we assume wlog that~$(c_1,\ldots, c_{|\c'|})$ is a topological ordering of~$\mathcal{H}(G_i)$ for all~$i \in [t]$.
 
 We now show a composition of the $t$ instances of \mca into a new instance of \mca. Let $G  = (V, A)$ be the graph of such a new instance with $V = \{r\} \cup \{r'_i : i \in [t]\} \cup \{ V_i : i \in [t]\}$ and $A = \{ (r, r'_i) : i \in [t]\} \cup \{ (r'_i, r_i) : i \in [t] \} \cup \{ A_i : i \in [t]\}$. Here,~$r$ is a vertex not contained in any of the $t$ \mca instances and which has a path of length 2 towards the root $r_i$ of any graph $G_i$ ; thus $G$ is clearly a DAG. Let $\c$ be the color set of $G$, and let us define the coloring function on $V(G)$ as follows: the root $r$ is assigned a unique color $c_r \notin \c'$ ; all vertices of type $r'_i$ are assigned the same color $c_{r'} \notin (\c' \cup \{c_r\})$ ; all arcs of type $(r'_i, r_i)$ and $(r, r'_i)$ are given a weight of 0 ; the color (resp. weight) of all other vertices (resp. arcs) is the same in the new instance than in their initial instance. Clearly, $(G, \c, col, w, r)$ is a correct instance of \mca. Moreover, if $G_i$ is an arborescence for every $i \in [t]$, then $G$ is also an arborescence.
We now prove that there exists $i \in [t]$ such that $G_i$ has a colorful arborescence $T = (V_T, A_T)$ rooted in $r_i$ of weight $W > 0$ if and only if $G$ has a colorful arborescence $T' = (V_{T'}, A_{T'})$ rooted in $r$ and of weight $W > 0$.  

($\Rightarrow$) If there exists $i \in [t]$ s.t. $G_i$ has a colorful arborescence $T = (V_T, A_T)$ rooted in $r_i$ and of weight $W > 0$, then let $T' = (V_{T'}, A_{T'})$ with $V_{T'} = V_T \cup \{r, r'_i\}$ and $A_{T'} = A_T \cup \{(r, r'_i),(r'_i, r_i)\}$. Clearly, $T'$ is connected, colorful and of weight $W$. 

($\Leftarrow$) Suppose there exists a colorful arborescence $T' = (V_{T'}, A_{T'})$ rooted in $r$ in $G$ of weight $W > 0$. Since $T'$ is colorful and all vertices of type $r'_i$ share the same color, there cannot exist $i$ and $j$ in $[t]$, $v_i \in V_i$ and $v_j \in V_j$ such that both $v_i$ and $v_j$ belong to $T'$. Thus, let $i^*$ be the only index in $[t]$ such that $V_{i^*} \cap V_{T'} \neq \emptyset$ and let $T = (V_T, A_T)$ with $V_T = V_{T'} \setminus \{r, r'_{i^*} \}$ and $A_T = A_{T'} \setminus \{(r, r'_{i^*}), (r'_{i^*}, r_{i^*})\}$. Clearly, $T$ is connected, colorful and of weight $W$.

Now, notice that $|\c| = |\c'|+2$ and thus that we made a correct composition of \mca into \mca. Moreover, recall that \mca is $\np$-hard~\cite{DBLP:conf/tamc/FertinFJ17} and that $\nhs \leq |\c|$. As a consequence, \mca does not admit a polynomial kernel relatively to $|\c|$, and hence relatively to $\nhs$, even in arborescences, unless $\np \subseteq \cnpp$.  
\end{proof}

Recall that \mca can be solved in time $\bigos(2^s)$ where~$s$ is the minimum number of arcs needet to turn $\h$ into an arborescence~\cite{DBLP:conf/tamc/FertinFJ17}. Since $s <|\c|^2$, we have the following. %of Proposition~\ref{prop:no-poly-kernel-c}.
\begin{corollary}
\label{cor:no-poly-kernel-s}
Unless $\np \in \cnpp$, \mca does not admit a polynomial kernel relatively to $s$, even if~$G$ is an arborescence.
\end{corollary}

%%%%%%%%%%%%%%%%%%%%%%%%%%%%%%%%%%%%%%%%%%%%%%%%%%%%%%%%%%%%%%%%%%%%%%%%%%%%%%%%%%%%%%%%%%%%%%%%%%%%%%%%%%%%%%%
%\subsection*{Reduction from \sc}

%In Proposition~\ref{prop:no-poly-kernel-c}, we made a composition of \mca into \mca in order to show that \mca is unlikely to admit a polynomial kernel relatively to $|\c|$, and hence to $\nhs$. 

In the following, we use a different technique, called polynomial parameter transformation~\cite{DBLP:journals/tcs/BodlaenderTY11}, to show that \mca is also unlikely to admit a polynomial kernel relatively to $\nhs + \lc$, where $\lc=|V(G)|-|\c|$. %Let $P$ be an $\np$-hard problem and $Q$ a problem which belongs to $\np$. Roughly speaking, a polynomial parameter transformation is a reduction of $P$ parameterized by $k$ into $Q$ parameterized by $k'$ such that $k'=p(k)$, where $p$ is a polynomial. Therefore, if $P$ parameterized by $k$ admits a \emph{polynomial parameter transformation} into $Q$ parameterized by $k'$ and if $Q$ admits a polynomial kernel relatively to $k'$, then $P$ admits a polynomial kernel relatively to $k$~\cite{DBLP:journals/tcs/BodlaenderTY11}. Using such a transformation, we obtain the following result.

\begin{proposition}
\label{prop:no-poly-kernel-lc-in2}
\mca does not admit any polynomial kernel relatively to $\nhs + \lc$ unless $\np \subseteq \cnpp$.
\end{proposition}

Since $\ell \geq \lc$, and in light of Proposition~\ref{prop:no-poly-kernel-lc-in2}, we aim at determining whether a polynomial kernel exists for \mca relatively to $\nhs + \ell$. We have the following theorem.
\begin{theorem}
\label{thm:kernel-l-in2}
\mca admits a problem kernel with $\bigo(\nhs\cdot \ell^2)$ vertices.
\end{theorem}

%%%%%%%%%%%%%%%%%%%%%%%%%%%%%%%%%%%%%%%%%%%%%%%%%%%%%%%%%%%%%%%%%%%%%%%%%%%%%%%%%%%%%%%%%%%%%%%%%%%%%%%%%%%%%%%%%%%%%%%%%%%%%%%%%%%%%%%%%%

\section{Parameterizing the \mca Problem by $\tw$}
\label{sec:tw}
Let $U(\h)$ denote the underlying undirected graph of $\h$. In this section, we are interested in parameter $\tw$, defined as the treewidth of $U(\h)$. Indeed, since \mca is in $\p$ whenever $\h$ is an arborescence~\cite{DBLP:conf/tamc/FertinFJ17}, it is natural to study whether \mca parameterized by $\tw$ is $\fpt$. To do so, we first introduce some definitions. 

\begin{definition} %\textbf{\cite{Niedermeier2006}}
\label{def:tree-decomposition}
Let $G = (V,E)$ be a undirected graph. A tree decomposition of $G$ is a pair $\langle \{X_i : i \in I\}, \t \rangle$, where $\t$ is a tree whose vertex set is $I$, and each $X_i$ is a subset of $V$, called a \emph{bag}. The following three properties must hold: \\
1) $\cup_{i \in I} X_i = V$~; \\
2) For every edge $(u,v) \in E$, there is an $i \in I$ such that $\{u,v\} \subseteq X_i$~; \\
3) For all $i,j,k \in I$, if $j$ lies on the path between $i$ and $k$ in $\t$, then $X_i \cap X_k \subseteq X_j$. 
\end{definition}

The \emph{width} of $\langle \{X_i : i \in I\}, \t \rangle$ is defined as  $\textnormal{max}\{|X_i| : i \in I\}-1$, and the \emph{treewidth} of $G$ is the minimum $k$ such that $G$ admits a tree decomposition of width $k$.

\begin{definition} %\cite{Niedermeier2006}
\label{def:nice-tree-decomposition}
A tree decomposition $\langle \{X_i : i \in I\}, \t \rangle$ is called \emph{nice} if the following conditions are satisfied:\\
1) Every node of $\t$ has at most two children~;\\
2) If a node $i$ has two children $j$ and $k$, then $X_i = X_j = X_k$ and in this case, $X_i$ is called a \textsc{Join Node}~;\\
3) If a node $i$ has one child $j$, then one of the following situations must hold:\\
$~~~$ a) $|X_i| = |X_j| +1$ and $X_j \subset X_i$ and in this case, $X_i$ is called an \textsc{Introduce Node}, or \\
$~~~$ b) $|X_i| = |X_j| -1$ and $X_i \subset X_j$ and in this case, $X_i$ is called a \textsc{Forget Node}\\
4) If a node $i$ has no child, then $|X_i| = 1$ and in this case, $X_i$ is called a \textsc{Leaf Node}.
\end{definition}

%\subsection*{Reduction from \sc}
We first show in the next proposition that \mca is unlikely to be $\fpt$ with respect to parameter $\tw$.

\begin{proposition}
\label{prop:w2-tw}
\mca is $\wtwo$-hard relatively to $\tw$.
\end{proposition}

\begin{proof}
We reduce from the $k$-\msc problem, which is defined below.

\begin{center}
\fbox{\begin{minipage}{0.95\textwidth}
\textbf{$k$-\msc}\\
\textbf{Input:} A universe $\mathcal{U} = \{u_1, u_2, \ldots, u_q\}$, a family $\mathcal{F} = \{S_1, S_2, \ldots, S_p\}$ of subsets of $\mathcal{U}$, a set of colors $\Lambda$
%= \{\lambda_1, \lambda_2, \ldots, \lambda_k\}$ 
with a coloring function $\col' : \mathcal{F} \to \Lambda$, an integer $k$.\\
\textbf{Output:} A subfamily $\s \subseteq \mathcal{F}$ of sets whose union is $\mathcal{U}$, and such that (i)~$|\s|=k$ and (ii)~$\s$ is colorful, \textit{i.e.} $\col'(S_i)\neq col'(S_j)$ for any $i\neq j$ such that $S_i, S_j \in \s$. 
\end{minipage}}
\end{center}

The reduction is as follows: for any instance of $k$-\msc, we create a three-level DAG $G = (V = V_1 \cup V_2 \cup V_3, A)$ with $V_1 = \{r\}$, $V_2 = \{ v_i : i \in [p] \}$ and $V_3 = \{ z_j : j \in [q] \}$. Informally, we associate a vertex at the second level to each set of $\mathcal{F}$ and a vertex at the third level to each element of $\mathcal{U}$. We then add an arc of weight $-1$ from $r$ to each vertex at level 2 and an arc of weight $p$ from $v_i$ to $z_j$, for all $i \in [p]$ and $j \in [q]$ such that $u_j\in S_i$. Now, our coloring function $\col$ is as follows: we give a unique color to each vertex in $V_1 \cup V_3$, while at the second level (thus in $V_2$), two vertices of type $v_i$ are assigned the same color if and only if their two associated sets are assigned the same color by $\col'$. Notice that $\h$ is also a three-levels DAG with resp. $\col(V_1)$, $\col(V_2)$ and $\col(V_3)$ at the first, second and third levels. Therefore, $(G, \c, col, w, r)$ is a correct instance of \mca.
We now prove that there exists a colorful set $\s \in \mathcal{F}$ of size $k$ whose union is $\mathcal{U}$ if and only if there exists a colorful arborescence $T$ in $G$ of weight $w(T) = pq-k$. 

$(\Rightarrow)$ Suppose there exists a colorful set $\s \in \mathcal{F}$ of size $k$ whose union is $\mathcal{U}$ and let $\texttt{True} = \{i \in [p] : S_i \in \s \}$. Let $V_T = \{r\} \cup \{v_i : i \in \texttt{True} \} \cup \{ z_j : j \in [q] \}$. Necessarily, $G[V_T]$ is connected: first, $r$ is connected to every level-2 vertex ; second, a vertex $z_j$ corresponds to an element $u_j$ which is contained in some set $S_i \in \s$. Now, let $T$ be a spanning arborescence of $G[V_T]$. Clearly, $T$ is colorful and of weight $pq-k$. 

($\Leftarrow$) Suppose there exists a colorful arborescence $T=(V_T, A_T)$ in $G$ of weight $w(T) = pq-k$. Notice that any arborescence $T'$ in $G$ which contains $r$ and at least one vertex from $V_3$ must contain at least one vertex from $V_2$ in order to be connected. Therefore, if such an arborescence $T'$ does not contain one vertex of type $z_j$, then $w(T') < pq-p-1$ and $w(T') < w(T)$. Hence, if $w(T) = pq-k$ then $T$ necessarily contains each vertex from $V_3$, and thus contains exactly $k$ vertices from $V_2$. Now, let $\s = \{ S_i : i \in [p]~s.t.~v_i \in V_T\}$ and notice that $\s$ is a colorful subfamily of size $k$ whose union is $\mathcal{U}$ as all vertices of the third level belong to $T$. Our reduction is thus correct.

Now, recall that $\h$ is a three-levels DAG with resp. $\col(V_1)$, $\col(V_2)$ and $\col(V_3)$ at the first, second and third levels. Thus, there exists a trivial tree decomposition $\langle \{X_i : i \in [|\col(V_3)|+2]\}, \t \rangle$ of $U(\h)$ which is as follows: the bag $X_0 = \{col(r)\}$ has an arc towards the bag $X_1 = \{ \{\col(r)\} \cup \col(V_2)\}$ and, for any $i \in [|col(V_3)|]$, there exists an arc from $X_1$ to $X_i$ where each $X_i$ contains $\col(V_2)$ and a different vertex of $\col(V_3)$. Consequently, the width of $\langle \{X_i : i \in [|col(V_3)|+2]\}, \t \rangle$ is $k$, and hence $\mca$ is $\wtwo$-hard parameterized by $\tw$ as $k$-\msc is well-known to be $\wtwo$-hard parameterized by~$k$. 
\end{proof}

We now use the above proof to show that \mca is unlikely to admit $\fpt$ algorithms relatively for different further parameters related to $\h$. The vertex-cover number of $U(\h)$ is the size of a smallest subset $S \subseteq V(\h)$ such that at least one incident vertex of any arc of $\h$ belongs to $S$. Notice that $\col(V_2)$ is a vertex cover of $U(\h)$ and thus $U(\h)\le k$. The \emph{feedback vertex set} number is the size of a smallest subset $S \subseteq \h$ whose removal makes $U(\h)$ acyclic. The size of such a subset $S$ is an interesting parameter as $\nhs = 0$ in $\h[V(\h)\setminus S]$ and any vertex cover of $U(\h)$ is also a feedback vertex set of $U(\h)$ -- hence, $\col(V_2)$ is also a feedback vertex set of $U(\h)$. Altogether, we thus obtain the following corollary. 

\begin{corollary}
\label{cor:msc-implies-vc}
\mca is $\wtwo$-hard relatively to the vertex-cover number of $U(\h)$, and relatively to the feedback vertex set number of $U(\h)$. 
\end{corollary}

Next, recall that, in proof of in Proposition~\ref{prop:w2-tw}, each color from the third level of $\h$ is a leaf. Hence, the number of colors of outdegree at least $2$ in $\h$ is $|\col(V_1)|+|\col(V_2)|=k+1$. Although Theorem~\ref{thm:fpt-in2} showed that \mca is $\fpt$ relatively to $\nhs$, we obtain the following.

\begin{corollary}
\label{cor:msc-implies-out2}
\mca is $\wtwo$-hard relatively to the number of colors of outdegree at least $2$ in $\h$.
\end{corollary}

%%%%%%%%%%%%%%%%%%%%%%%%%%%%%%%%%%%%%%%%%%%%%%%%%%%%%%%%%%%%%%%%%%%%%%%%%%%%%%%%%%%%%%%%%%%%%%%%%%%%%%%%%%%%%%%%%%%%%%%%%%%%%%%%%%%%%%%%%%

%\subsection*{FPT algorithm for $\lc+\tw$}

By Proposition~\ref{prop:w2-tw}, \mca parameterized by~$\tw$ is $\wtwo$-hard~; thus, one may look for a parameter whose combination with~$\tw$ may lead to \mca being $\fpt$. Here, we focus on parameter $\lc = n_G - |\c|$. We know that \mca is $\wone$-hard relatively to $\lc$, but the problem can be solved in $\bigos(2^{\lc})$ when $G$ is an arborescence~\cite{DBLP:conf/tamc/FertinFJ17}. Recall also that \mca is in $\p$ when $\h$ is an arborescence~\cite{DBLP:conf/tamc/FertinFJ17}, and hence when $\tw = 1$. 
%Although the fact that $G$ is an arborescence does not necessarily imply that $\h$ is also an arborescence, determining whether \mca is $\fpt$ relatively to $\lc+\tw$ remains of interest. 
In the following, a \emph{fully-colorful subgraph of $G$} is a subgraph of $G$ that contains {\em exactly} one occurrence of each color $c \in \c$.
\begin{lemma}
\label{lem:2-lc-colorful}
Any graph $G$ with $|\c|$ colors has at most $2^{\lc}$ fully-colorful subgraphs.
\end{lemma} 

\begin{proof}
Let $n_c$ be the number of vertices of color $c \in \c$ and notice that $\prod_{c \in \c}{n_c}$ is the number of fully-colorful subgraphs of $G$. Then, observe that $n_c \le 2^{n_c -1}$ for all $n_c \in \mathbb{N}$, which implies $\prod_{c \in \c}{n_c} \le 2^{ \sum_{c \in \c} n_c -1 }$ and thus $\prod_{c \in \c}{n_c} \le 2^{\lc}$. 
\end{proof}

%We are now ready to prove the following theorem.

\begin{theorem}
\label{thm:fpt-lc-twch}
\mca can be solved in $\bigos(2^{\lc}\cdot 4^{\tw})$ time and $\bigos(3^{\tw})$ space.
\end{theorem}

\begin{proof}
In the following, let $\langle \{X_i : i \in I\}, \t \rangle$ be a nice tree decomposition of $U(\h)$. In this proof, we provide a dynamic programming algorithm that makes use of $\langle \{X_i : i \in I\}, \t \rangle$ in order to compute a solution to \mca in any fully-colorful subgraph $G' \subseteq G$, to which we remove all vertices that are not accessible from $r$.
First, observe that $\langle \{X_i : i \in I\}, \t \rangle$ is also a correct nice tree decomposition for the (undirected) color hierarchy graph of any subgraph of $G$. Second, as any colorful graph is equivalent to its color hierarchy graph, notice that $\langle \{X_i : i \in I\}, \t \rangle$ is also a correct nice tree decomposition of any fully-colorful subgraph $G' \in G$. Therefore, we assume wlog that any bag $X_i$ contains vertices of such graph $G'$ instead of colors, and that $X_0 = \{r\}$ is the root of $\langle \{X_i : i \in I\}, \t \rangle$. 

Now, for any $i \in I$ and for any subsets $L_1, L_2, L_3$ that belong to $X_i$ such that $L_1 \oplus L_2 \oplus L_3 = X_i$, let $T_i[L_1, L_2, L_3]$ store the weight of a \emph{partial solution} of \mca in $G'$, which is a collection of $|L_1|$ disjoint arborescences such that : 
\begin{itemize}
 \item each $v \in L_1$ is the root of exactly one such arborescence,
 \item each $v \in L_2$ is contained in exactly one such arborescence,
 \item no vertex~$v \in L_3$ belongs to any of these arborescences,
 \item any vertex $v \in V$ whose color is forgotten below $X_i$ can belong to any such arborescence,
 \item there does not exist another collection of arborescences with a larger sum of weights under the same constraints.
\end{itemize}
Besides, let us define an entry of type $D_i[L_1, L_2, L_3]$ which stores the same partial solution as entry $T_i[L_1, L_2, L_3]$, except for the vertices $v \in V$ whose colors are forgotten below $X_i$ which cannot belong to any arborescence of the partial solution. We now detail how to compute each entry of $T_i[L_1, L_2, L_3]$. We stress that each entry of $D_i[L_1, L_2, L_3]$ is filled exactly as an entry of type $T_i[L_1, L_2, L_3]$, apart from the case of forget nodes which we detail below.\\

\noindent \textit{$\bullet$ If $X_i$ is a leaf node :} $T_i[L_1 , L_2, L_3]$ = 0%\]

\noindent Notice that leaf nodes are base cases of the dynamic programming algorithm as $\langle \{X_i : i \in I\}, \t \rangle$ is a nice tree decomposition. Moreover, recall that leaf nodes have size 1 and thus that the only partial solution for such nodes has a weight of zero. \\

\noindent \textit{$\bullet$ If $X_i$ is an introduce node having a child $X_j$ and if $v^*$ is the introduced vertex :} \\
\[
T_i[L_1, L_2, L_3] = 
\left\{ 
\begin{array}{l l}
  A)~\underset{\forall S \subseteq L_2}{\max} \{ \underset{v \in S}{\sum}w(v^*, v) + T_j[L_1 \cup S \setminus \{v^*\}, L_2 \setminus S, L_3]\} \\
  \phantom{C)~T_j[L_1, L_2, L_3 \setminus \{v^*\}]\} \qquad \qquad \qquad \qquad} \text{if $v^* \in L_1$} \\
  B)~\underset{\forall u \in (L_1 \cup L_2)}{\max} \{ w(u, v^*) + \\
  \quad \quad \underset{\forall S \subseteq (L_2 \setminus \{u\})}{\max} \{ \underset{v \in S}{\sum}w(v^*, v) + T_j[L_1 \cup S \setminus \{v^*\}, L_2 \setminus S, L_3]\} \} \\ 
  \phantom{C)~T_j[L_1, L_2, L_3 \setminus \{v^*\}]\} \qquad \qquad \qquad \qquad} \text{if $v^* \in L_2$} \\
  C)~T_j[L_1, L_2, L_3 \setminus \{v^*\}]\} \qquad \qquad \qquad \qquad \text{if $v^* \in L_3$} 
\end{array} 
\right.
\]

\noindent where we set $w(u,v) = - \infty$ when there is no arc from $u$ to $v$ in $G'$. There are three cases: $v^*$ is the root of an arborescence in a partial solution (case $A)$), an internal vertex of such a solution (case $B)$) or ~$v^*$ does not belong to such a solution (case $C)$). In case $A)$, $S$~corresponds to the set of outneighbors of~$v^*$ in the partial solution, thus the vertices of~$S$ do not have any other inneighbor in the partial solution. Therefore, in the corresponding entry $T_j$, the vertices of~$S$ are roots. Now, notice that $B)$ is very similar to $A)$. In addition to a given set $S$ of outneigbors, $v^*$ being in $L_2$ implies that $v^*$ has an inneighbor~$u \in (L_1 \cup L_2)$ in the partial solution. Since the inneighbor~$u$ cannot be an outneighbor at the same time, $u$~is not contained in~$S$. Exhaustively trying all possibilities for both~$S$ and~$u$ ensures the correctness of the solution. Finally, by definition of $L_3$, observe that $v^*$ does not belong to the partial solution of $T_i[L_1, L_2, L_3]$ if $v^* \in L_3$. \\

\noindent \textit{$\bullet$ If $X_i$ is a forget node having a child $X_j$ and if $v^*$ is the forgotten vertex :}
%\[
$$T_i[L_1, L_2, L_3] = \max \{ T_j[L_1 , L_2 \cup \{v^*\}, L_3], T_j[L_1 , L_2, L_3 \cup \{v^*\}] \}$$ %\quad \\
%\]

\begin{comment}
\noindent First, by definition of a treewidth decomposition, for any $(v,v^*) \in A$ (resp. $(v^*,v) \in A$), there exists at least one index $i \in [I]$ such that the corresponding bag $X_i$ contains both $v$ and $v'$. Therefore, any combination of vertices and arcs related to $v^*$ in order to determine whether $v^*$ belongs to the partial solution or not has already been tried out when $v^*$ is forgotten. Second, for any vertex $v \in V$, such a definition implies that there exists a path from the forgotten node of $v$ to any node which contains $v$ in $\langle \{X_i : i \in I\}, \t \rangle$. Therefore, notice that $X_i$ necessarily contains at least one vertex $u \in V(G)$ such that there exists a path from $\col(u)$ to $\col(v^*)$ in $\h$, otherwise any vertex $v \in V(G)$ that belongs to the path from $r$ (included) to $v^*$ in $G$ is already forgotten in $\langle \{X_i : i \in I\}, \t \rangle$, and hence this clearly contradicts the fact that $X_o = \{r\}$. Finally, observe that every entry $T_j[L_1 \cup \{v^*\},L_2,L_3]$ does not matter in order to compute $T_i[L_1,L_2,L_3]$, and hence the formula determines whether the collection of arborescences that is stored in $T_i[L_1,L_2,L_3]$ had a higher weight with or without $v^*$. \\
\end{comment}

\noindent Informally, the above formula determines whether the collection of arborescences that is stored in $T_i[L_1,L_2,L_3]$ had a higher weight with or without $v^*$ as an internal vertex. Observe that we do not consider the case where $v^*$ is the root of an arborescence as such an arborescence could not be connected to the rest of the partial solution via an introduced vertex afterwards. Besides, notice that $D_i[L_1, L_2, L_3] = D_j[L_1 , L_2, L_3 \cup \{v^*\}]$ as the partial solution in $D_i[L_1, L_2, L_3]$ does not contain any forgotten vertex by definition. \\

\noindent \textit{$\bullet$ If $X_i$ is a join node having two children $X_j$ and $X_k$ :}
$$T_i[L_1, L_2, L_3] = T_j[L_1 , L_2, L_3]~+~T_k[L_1 , L_2, L_3]-D_i[L_1 , L_2, L_3]$$

\noindent Informally, the partial solution in $T_i[L_1 , L_2, L_3]$ can contain both the forgotten vertices of the partial solution in $T_j[L_1 , L_2, L_3]$ and those of the partial solution in $T_k[L_1 , L_2, L_3]$. Recall that the partial solution in $D_i[L_1 , L_2, L_3]$ does not contain any forgotten vertices and therefore that any arc of the partial solution in $T_i[L_1 , L_2, L_3]$ is only counted once.

\medskip

\noindent We fill the tables from the leaves to the root for all $i \in I$ until $T_0$ and any entry of type $T_i[L_1 , L_2, L_3]$ is directly computed after the entry of type $D_i[L_1 , L_2, L_3]$. If $T' = (V_{T'}, A_{T'})$ is a solution of \mca in a fully-colorful subgraph $G' \subseteq G$, then $w(T') = T_0[\{r\}, \emptyset, \emptyset]$. Thus, for each fully-colorful subgraph we can compute the solution by filling the tables~$T$ and~$D$. The table has~$3^{\tw}$ entries which implies the upper bound on the space consumption. The most expensive recurrences in terms of running time are the one of cases~A) and~B) for introduce nodes~$X_i$ where we consider altogether~$\bigo(4^{\tw})$ cases: each term corresponds to a partition of~$X_i$ into four sets~$L_1$,~$L_2\setminus S$,~$L_2\cap S$, and~$L_3$. Finally, the solution of \mca in $G$ is also the solution of at least one fully-colorful subgraph $G' \subseteq G$. Therefore, computing the solution of \mca for any such subgraph $G'$ ensures the correctness of the algorithm and hence, by Lemma~\ref{lem:2-lc-colorful}, adding a factor $\bigo(2^{\lc})$ to the complexity of the above algorithm proves our theorem. 
\end{proof}

\newpage
%\section{Conclusion}
%\label{sec:concl}

%In this paper, we obtained a $\bigos(3^{\nhs})$ time algorithm, which improves upon the $\bigos(3^{|\c|})$ of Böcker \textit{et al.}~\cite{DBLP:conf/eccb/BockerR08}. We also showed that \mca is unlikely to admit a polynomial kernel relatively to $\nhs + \lc$ and then that the problem admits such a kernel relatively to $\nhs + \ell$. Furthermore, we proposed a $\fpt$ algorithm for \mca relatively to $\tw + \lc$, although we showed that \mca is $\wtwo$-hard relatively to $\tw$. \\
%In light of these results, we ask the following question: does \mca admit a polynomial kernel relatively to $\lc + \tw$ ?

%%%%%%%%%%%%%%%%%%%%%%%%%%%%%
%%% BIBLIOGRAPHY
%%%%%%%%%%%%%%%%%%%%%%%%%%%%%
\bibliography{references}
%%%%%%%%%%%%%%%%%%%%%%%%%%%%%

%%%%%%%%%%%%%%%%%%%%%%%%%%%%%
%%% APPENDIX
%%%%%%%%%%%%%%%%%%%%%%%%%%%%%
\newpage
\section*{Appendix}
%\label{sec:appendix}

\appendix
\label{sec:appendix}

\section{Proofs from Section~\ref{sec:in2}}

\begin{comment}
\renewcommand{\thetheorem}{\ref{lem:in2-disjoint-colors}}

%\setcounter{lemma}{0}  % reset counter
\begin{lemma}
%\label{lem:in2-disjoint-colors}
In $\h$, for any $c \in \c$, any pair of distinct colors $c_1, c_2 \in N^+(c)$ and any disjoint sets $X_1, X_2 \subseteq X$, any arborescence $T_1$ rooted in $c_1$ s.t. $X(T_1) \subseteq X_1$ is disjoint from any arborescence $T_2$ rooted in $c_2$ s.t. $X(T_2) \subseteq X_2$.
\end{lemma}
\end{comment}

\renewcommand{\thetheorem}{\ref{prop:no-poly-kernel-lc-in2}}
\begin{proposition}
\mca does not admit any polynomial kernel relatively to $\nhs + \lc$, unless $\np \subseteq \cnpp$.
\end{proposition}

We first present the formal definition of a polynomial parameter transformation and then use that technique to prove the proposition.

\renewcommand{\thetheorem}{2.8}
\begin{definition} (\cite{DBLP:journals/tcs/BodlaenderTY11,DBLP:journals/talg/DomLS14,DBLP:journals/dam/CyganPPW12})
\label{def:poly-param-transfo}
Let $P$ and $Q$ be two parameterized problems. We say that $P$ is polynomial parameter reducible to $Q$ if there exists a polynomial-time computable function $f : \Sigma^* \times \mathbb{N} \to \Sigma^* \times \mathbb{N}$ and a polynomial $p$, such that for all $(x,k) \in \Sigma^* \times \mathbb{N}$ the following holds: $(x,k) \in P$ iff $(x',k')=f(x,k) \in Q$, and $k' \leq p(k)$.
The function $f$ is a called a \emph{polynomial parameter transformation}.
\end{definition}

\begin{proof}
We reduce from \sc, which is defined as follows.

\begin{center}
\fbox{\begin{minipage}{0.95\textwidth}
\textbf{\sc}\\
\textbf{Input:} A universe $\mathcal{U} = \{u_1, u_2,\ldots, u_q\}$, a family $\mathcal{F} = \{S_1, S_2,\ldots, S_p\}$ of subsets of $\mathcal{U}$, an integer $k$.\\
\textbf{Output:} A $k$-sized subfamily $\s \subseteq \mathcal{F}$ of sets whose union is $\mathcal{U}$. 
\end{minipage}}
\end{center}

The reduction is as follows: for any instance of \sc, we create a three-levels DAG $G = (V = V_1 \cup V_2 \cup V_3, A)$ with $V_1 = \{r\}$, $V_2 = \{ v_i : i \in [p] \}$ and $V_3 = \{ z_j : j \in [q] \}$. We call $V_2$ the second level of~$G$ and~$V_3$ the third level of~$G$. Informally, we associate one vertex at the second level to each set of $\mathcal{F}$ and one vertex at the third level to each element of $\mathcal{U}$. There is an arc of weight $-1$ from $r$ to each vertex at level 2 and an arc of weight $p$ from $v_i$ to $z_j$, for all $i \in [p]$ and $j \in [q]$ such that the element $u_j$ is contained in the set $S_i$. Now, our coloring function $\col$ is as follows: give a unique color to each vertex of $G$. Notice that $\h$ is also a three-levels DAG with resp. $\col(V_1)$, $\col(V_2)$ and $\col(V_3)$ at the first, second and third levels.
Therefore, the above construction is a correct instance of \mca.
We now prove that there exists a $k$-sized subfamily $\s \subseteq \mathcal{F}$ of sets whose union is $\mathcal{U}$ if and only if there exists a colorful arborescence $T$ in $G$ of weight $w(T) = pq-k$. 

$(\Rightarrow)$ Suppose there exists a $k$-sized subfamily $\s \subseteq \mathcal{F}$ of sets whose union is $\mathcal{U}$ and let $\texttt{True} = \{i \in [p] : S_i \in \s \}$. Then, we set $V_T = \{r\} \cup \{v_i : i \in \texttt{True} \} \cup \{ z_j : j \in [q] \}$. Necessarily, $G[V_T]$ is connected: first, $r$ is connected to every level-2 vertex~; second, a vertex $z_j$ corresponds to an element $u_j$ which is contained in some set $S_i \in \s$. Now, let $T$ be a spanning arborescence of $G[V_T]$. Clearly, $T$ is colorful and of weight $pq-k$. 

($\Leftarrow$) Suppose there exists a colorful arborescence $T=(V_T, A_T)$ in $G$ of weight $w(T) = pq-k$. Notice that any arborescence $T'$ in $G$ which contains $r$ and at least one vertex from $V_3$ must contain at least one vertex of from $V_2$ in order to be connected. Therefore, if such an arborescence $T'$ does not contain one vertex of type $z_j$, then $w(T') < pq-p-1$ and $w(T') < w(T)$. Hence, if $w(T) = pq-k$ then $T$ contains each vertex of the third level and $T$ contains exactly $k$ vertices at the second level. Now, let $\s = \{ S_i : i \in [p]~s.t.~v_i \in V_T\}$ and notice that $\s$ is a $k$-sized subfamily of $\mathcal{F}$ whose union is $\mathcal{U}$ as all vertices of the third level belong to $T$. Our reduction is thus correct.

Now, recall that $\h$ is a three-levels DAG with resp. $\col(V_1)$, $\col(V_2)$ and $\col(V_3)$ at the first, second and third levels. By construction of $G$, if there exists $c \in V(\h)$ such that $d^-(c) \geq 1$, then $c \in col(V_3)$. Moreover, recall that $|col(V_3)| = |\u|$ and observe that $\lc = 0$ as $G$ is colorful. Therefore, $\nhs + \lc \leq |\u|$ and we provided a correct polynomial parameter transformation from \sc parameterized by~$|\u|$ to \mca parameterized by~$\nhs + \lc$. Now, recall that \sc is unlikely to admit a polynomial kernel for~$|\u|$~\cite{DBLP:journals/talg/DomLS14} and that \sc is $\np$-hard~\cite{DBLP:conf/coco/Karp72}. Moreover, the decision version of \mca, which asks for a solution of weight at least $k$, clearly belongs to $\np$. As a consequence, \mca does not admit any polynomial kernel for~$\nhs + \lc$ unless~$\np \subseteq \cnpp$. 
\end{proof}

%%%%%%%%% THEOREM KERNEL IN2 + L
\paragraph*{Proof of the Polyonomial Kernel for $\nhs + \ell$}
Recall that $X$ is the set of \textit{difficult colors} which are the colors of indegree at least 2 in $\h$, and hence $|X| = \nhs$. 
\renewcommand{\thetheorem}{\ref{thm:kernel-l-in2}}
\begin{theorem}
\mca admits a problem kernel with $\bigo(\nhs\cdot \ell^2)$ vertices.
\end{theorem} 
To show this result we provide two data reduction rules. To formulate the rules, we introduce some notation first.  

 %The indegree of any vertex $v \in V(G)$ (resp. color $c \in V(\h)$) is denoted $d^-(v)$ (resp. $d^-(c)$). 
We say that a vertex $v \in V(G)$ (resp. $c \in V(\h)$) is \emph{reachable} from another vertex $v' \in V(G)$ (resp. $c' \in V(\h)$) if there exists a path from $v'$ to $v$ in $G$ (resp. from $c'$ to $c$ in $\h$). For any vertex $v \in V(G)$, we define $G^+(v)$ as the induced subgraph of the set of vertices that are reachable from $v$ in $G$ (including $v$). Similarly, for any color $c \in V(\h)$, we define $\h^+(c)$ as the induced subgraph of the set of vertices that are reachable from $c$ in $\h$ (including $c$). We call such a color $c$  \emph{autonomous} if %and only if
$(i)$ $\h^+(c)$ is an arborescence and $(ii)$ there does not exist an arc from a color $c_1 \notin \h^+(c)$ to a color $c_2 \in \h^+(c)$ in $\h$. For a vertex $v$, let $T_v$ denote the the maximum colorful arborescence which is rooted at $v$ in $G$. Finally, for a color $c \in \c$, we let $V_c := \{ v \in V : \col(v) = c\}$ denote the set of vertices with color~$c$.

\begin{rrule}\label{rule:autonomous-color}
  If an instance $(G, \c, col, w, r)$ of \mca contains an autonomous color~$c$ such that $\h^+(c)$ contains at least two vertices, then do the following.
  \begin{itemize}
  \item For each vertex~$v\in V_c$, compute the value~$w(T_v)$ of $T_v$, and add~$w(T_v)$ to the weight of each incoming arc of $v$.
  \item Remove from~$G$ all vertices that are reachable from a vertex in~$V_c$, except the vertices of~$V_c$.
  \end{itemize}
\end{rrule}

\renewcommand{\thetheorem}{2.9}
\begin{lemma}
Reduction Rule~\ref{rule:autonomous-color} is correct and can be performed exhaustively in polynomial time.
\label{lem:delete-poly-parts}
% For any instance $(G, \c, col, w, r)$ of \mca, if there exists a color $c \in \c$ such that $\h^+(c)$ is an arborescence, then there exists an equivalent instance $(G', \c', col', w', r)$ of \mca such that $\c' = \c \setminus V(\h^+(c)) \cup \{c\}$.
\end{lemma}

\begin{proof} 
Consider a vertex $v \in V_c$. Since $c$ is autonomous, $\h^+(c)$ is an arborescence and thus we may compute $T_v$ which contains only colors from $\h^+(c)$ in polynomial time~\cite{DBLP:conf/tamc/FertinFJ17}. 

Now, we prove the correctness of the rule, that is, the original instance has a colorful arborescence $T = (V_T, A_T)$ of weight $W$  if and only if the new instance has a colorful arborescence  $T' = (V_{T'}, A_{T'})$ of weight $W$. We show the forward direction of the equivalence; the converse can be seen by symmetric arguments. First, recall that $c$ is an autonomous color. Therefore, if $T$ does not contain any vertex of color $c$, then $T$ does not contain any vertex whose color belongs to $V(\h^+(c))$ and we can trivially set $T' = T$. Second, if $T$ contains a vertex $v$ of color $c$ and whose inneighbor is called $v^-$ in $T$, then we define a subset of vertices $S_c \subseteq V_T$ such that any vertex $v^+\in V_T$ belongs to $S_c$ if $v^+$ is reachable from $v$. We thus state that $V_{T'} = (V_T  \setminus S_c) \cup \{v\} $ and that $A_{T'}$ contains all the arcs from $A_T$ that are not in $\h^+(c)$. Now, recall that we computed the weight $w(T_{v})$ of the maximum colorful arborescence that was rooted in $v$ in $G$ and that $w'(v^-, v) = w(v^-, v) + w(T_{v})$, which ensures that $w(T) = w(T')$.  
\end{proof}

\renewcommand{\thetheorem}{2.10}
In the following, for any vertices $v, v' \in V(G)$ such that $v'$ is reachable from $v$ in $G$, we denote $\pi(v, v')$ as the length of the maximum weighted path from $v$ to $v'$ in $G$.
\begin{rrule}\label{rule:unique-inneighbor}
  If an instance $(G, \c, col, w, r)$ of \mca contains  a triple $\{c_1, c_2, c_3\} \subseteq \c$ such that $(i)~c_1$ is the unique inneighbor of $c_2$, $(ii)~c_2$ is the unique inneighbor of $c_3$ and $(iii)~c_3$ is the unique outneighbor of $c_2$, then do the following. 
  \begin{itemize}
  \item  For any~$v_1 \in V_{c_1}$ and $v_3 \in V_{c_3}$ such that there exists a path from~$v_1$ to~$v_3$ in $G$, create an arc $(v_1, v_3)$ and set~$w'(v_1, v_3) = \pi(v_1, v_3)$. 
  \item Add a vertex $v^*$ of color $c_3$ and, for any vertex $v_1 \in V_{c_1}$ that has at least one outneighbor of color~$c_2$ in $G$, add the arc $(v_1, v^*)$ and set $w'(v_1, v^*)$ to the highest weighted outgoing arc from $v_1$ to any vertex of color $c_2$ in $G$.
  \item Remove all vertices of~$V_{c_2}$ from~$G'$.
  \end{itemize}
\end{rrule}

\begin{lemma}
\label{lem:delete-lines}
Reduction Rule~\ref{rule:unique-inneighbor} is correct and can be performed exhaustively in polynomial time.
\end{lemma} 

\begin{proof} 
We first prove that our transformation is correct. We show only the direction that an arborescence of weight~$W$ in the original instance implies an arborescence of weight at least~$W$ in the new instance; the converse direction can be shown by symmetric arguments. Let $T = (V_T, A_T)$ be a colorful arborescence of weight $W$ in the original instance. First, if $T$ does not contain a vertex of color $c_2$, then $T$ is an arborescence of the new instance. Second, if $T$ contains a vertex $v_2$ of color $c_2$ whose inneighbor is $v_1$ in $T$ and if $T$ does not contain any vertex of color $c_3$, then setting $V_{T'} := V_T \setminus \{v_2\} \cup \{v^*\}$ and $A_{T'} := A_T \setminus \{(v_1, v_2)\} \cup \{(v_1, v^*)\}$ gives an arborescence~$T'=(V_{T'},A_{T'})$ of the new instance. Moreover, $w(T) = w'(T')$ since $w(v_1, v_2) = w'(v_1, v^*)$. Third, if $T$~contains a vertex~$v_2$ of color~$c_2$ whose inneighbor is~$v_1$ in $T$ and if $T$~contains a vertex~$v_3$ of color~$c_3$ (whose inneighbor is necessarily $v_2$), then setting $V_{T'} := V_T \setminus \{v_2\}$ and~$A_{T'} := A_T \setminus \{(v_1, v_2), (v_2, v_3)\} \cup \{(v_1, v_3)\}$ gives an arborescence~$T'=(V_{T'},A_{T'})$ of the new instance. Moreover, $w(T) = w'(T')$ since $w(v_1, v_2)+w(v_2,v_3) = w'(v_1, v_3)$. 

The polynomial running time follows from the fact that~$\pi(v_1, v_3)$ can be computed in polynomial time.
\end{proof}
To describe the final rule, let~$N_U^-(v)$ denote the set of unique colors in the inneighborhood of $v$ in $G$, where a color~$c$ is \textit{unique} if~$|V_c|=1$. Recall also that $\ell$ is the maximum number of vertices that do not belong to $T$ in $G$.
\begin{rrule}\label{rule:many-inneighbors}
  If an instance $(G, \c, col, w, r)$ of \mca contains a vertex~$v \in V$ such that $|N_U^-(v)| > \ell + 1$, then delete the $|N_U^-(v)| - \ell - 1$ least-weighted arcs from $N_U^-(v)$ to~$v$.
\end{rrule}
\begin{lemma}
\label{lem:inneighbors}
Reduction Rule~\ref{rule:many-inneighbors} is correct and can be performed exhaustively in polynomial time.
\end{lemma} 
\begin{proof}
   Since~$|N_U^-(v)| > \ell + 1$, $T$~has to contain at least two vertices from $N_U^-(v)$. Now, let $v_1$ be a vertex from $N_U^-(v)$ such that $(v_1,v)$ is the least-weighted incoming arc from a unique color to $v$ in $G$. Even if $v_1$ belongs to $T$, there will always exist at least one other vertex $v_2$ that will also belong to $T$ and such that $w(v_1, v) \le w(v_2, v)$. Thus, we may assume that~$T$ does not contain the arc $(v_1, v)$ and safely delete it. The correctness of the rule now follows from repeated application of this argument.
\end{proof}
%We are now ready to prove the following theorem.
We are now ready to prove Theorem~\ref{thm:kernel-l-in2}.

\begin{comment}
\begin{theorem}
\label{thm:kernel-l-in2}
\mca does admit a polynomial kernel relatively to $\nhs + \ell$.
\end{theorem} 
\end{comment}

\begin{proof}
 We first describe the kernelization process, then show that the obtained instance is bounded by a function of~$\nhs + \ell$.

 First, we iteratively reduce the input instance via Reduction
 Rules~\ref{rule:autonomous-color}--~\ref{rule:many-inneighbors}.  Let $(G, \c, \col, w, r)$
 denote the resulting instance which is equivalent and can be computed in polynomial time and
 let~$T = (V_T, A_T)$ be a solution of this instance.  First, we show that the indegree of any
 color in $\h$ is at most $(\ell +1)^2 + \ell$. This will allow us to show, in a second
 time, that $n_G$ is polynomially bounded by a function of $\nhs$ and $\ell$.

Let us first bound the indegree of any color in $\h$. Since $T$ is colorful and since $|V_T| = n_G - \ell$, there exists at most $\ell$ non-unique colors in $\c$ and hence the inneighborhood of any color $c \in V(\h)$ cannot contain more than $\ell$ non-unique colors in $\h$. Moreover, recall that the inneighborhood of any vertex $v \in V(G)$ cannot contain more than $\ell +1$ vertices of unique color in $G$, and that $T$ cannot be colorful if there exists more than $\ell +1$ occurrences of any color in~$G$. Hence, we may assume~$|V_c|\le c$. As a consequence, for any color $c \in V(\h)$, the inneighborhood of $c$ cannot contain more than $|V_c|\cdot(\ell +1)=(\ell +1)^2$ unique colors in $\h$, and hence~$c$ has at most~$(\ell +1)^2 + \ell$ inneighbors.

Now, let $F$ be the forest whose vertex set is $\c_F = \c \setminus X$ and which contains each arc $(c, c')$ of $\h$ such that $\{c, c'\} \subseteq \c_F$. In the following, we successively bound the maximum number of leaves of~$F$, the maximum number of vertices of~$F$, of~$V(\h)$ and finally of~$V(G)$ relatively to $\ell$ and $\nhs$. First, recall that there does not exist any autonomous color $c \in \c$ to which Reduction Rule~\ref{rule:autonomous-color} applies. Thus, each leaf $c$ of~$\h$ is in fact a difficult color. Consequently, every leaf of~$F$ is in~$\h$ an inneighbor of a difficult color.  Since the maximum indegree of any color in $\h$ is at most $(\ell + 1)^2 + \ell$, the number of leaves in~$F$ is bounded by $\nhs((\ell + 1)^2 + \ell)$. Now, by Lemma~\ref{lem:delete-lines}, $\h$ does not contain any color which has a unique inneighbor and a unique outneighbor. As a consequence, $F$ has no internal vertices of degree two that are not inneighbors of a difficult color. Hence, the number of nonleafs of~$F$ that are not inneighbors of a difficult color is~$\bigo(\nhs\cdot \ell^2)$, and thus $|V(F)| = \bigo(\nhs\cdot \ell^2)$. Moreover, since $\c_F = \c \setminus X$, we have that $|\c| \leq \nhs + \bigo(\nhs\cdot \ell^2)$. Finally, the number of vertices in~$G$ can exceed the number of colors in~$\h$ by at most~$\ell$. Therefore, $|V(G)|=\bigo(\nhs\cdot \ell^2)$ as claimed. 
\end{proof}

\end{document}